\documentclass[11pt]{article}

\usepackage{enumitem}
\usepackage{typearea}
\typearea{14}

\usepackage{setspace}

\usepackage{times}
\usepackage{color,epsfig}
\usepackage{amsmath,amsfonts,amssymb,amstext,amsthm}
\usepackage{xspace}
\usepackage{algorithm}
\usepackage{algpseudocode}
\usepackage{wrapfig}
\usepackage[compact]{titlesec}

\usepackage{thmtools,thm-restate} 
\newcommand{\nicefrac}{\frac}

\definecolor{Darkblue}{rgb}{0,0,0.4}
\definecolor{Brown}{cmyk}{0,0.61,1.,0.60}
\definecolor{Purple}{cmyk}{0.45,0.86,0,0}
\usepackage[backref,colorlinks,citecolor=Brown,bookmarks=true]{hyperref}

\allowdisplaybreaks

\newtheorem{theorem}{Theorem}[section]
\newtheorem{definition}[theorem]{Definition}
\newtheorem{lemma}[theorem]{Lemma}

\newtheorem{claim}[theorem]{Claim}
\newtheorem{remark}{Remark}
\newtheorem{fact}[theorem]{Fact}
\newtheorem{assumption}[theorem]{Assumption}

\numberwithin{algorithm}{section}

\newcommand{\OPT}{\ensuremath{\mathsf{OPT}}\xspace}
\newcommand{\ignore}[1]{}

\newcommand{\C}{\ensuremath{\mathbf{a}}\xspace}
\newcommand{\W}{\ensuremath{W}\xspace}

\newcommand*{\myproofname}{Proof of Claim}

\newcommand{\R}[0]{{\ensuremath{\mathbb{R}}}}

\newcommand{\poly}{\operatorname{poly}}

\newcommand{\initOneLiners}{
    \setlength{\itemsep}{0pt}
    \setlength{\parsep }{0pt}
    \setlength{\topsep }{0pt}
}

\newenvironment{OneLiners}[1][\ensuremath{\bullet}]
    {\begin{list}
        {#1}
        {\initOneLiners}}
    {\end{list}}

\usepackage{tikz}
\usetikzlibrary{decorations.pathmorphing}
\usetikzlibrary{decorations.markings}
\usetikzlibrary{arrows}

\def\adap{\ensuremath{{\sf adap}}\xspace}

\newcommand{\sse}{\subseteq}

\def\alg{\ensuremath{{\sf alg}}\xspace}

\def\E{\ensuremath{{\mathbb E}}}

\def\calH{\ensuremath{{\mathcal H}}}

\newcommand{\elt}{\ensuremath{ {\sf elt}}}

\newcommand{\depth}{deepness\xspace}

\def\smlc{\ensuremath{{\sf StocMLSC}}\xspace}

\newcommand{\T}{\ensuremath{\mathcal{T}}\xspace}

\newcommand{\one}[1]{\mathbf{1}{(#1)}}

\newcommand{\no}{\texttt{no}\xspace}
\newcommand{\yes}{\texttt{yes}\xspace}

\newcommand{\fmax}{f^{\max}}
\def\p{\mathbf{p}}

\DeclareMathOperator*{\argmax}{argmax}

\title{Adaptivity Gaps for Stochastic Probing:\\
Submodular and XOS Functions}
\author{
Anupam Gupta
\thanks{Computer Science Department, Carnegie Mellon
     University, Pittsburgh, PA 15213, USA. Research partly supported by
     NSF awards CCF-1319811, CCF-1536002, CCF-1540541 and CCF-1617790.}
 \and Viswanath Nagarajan
\thanks{Department of Industrial and Operations Engineering, University of Michigan, Ann Arbor, MI 48109, USA.}
\and Sahil Singla$^*$
}
\date{\today}

\begin{document}

\maketitle

\begin{abstract}
  Suppose we are given a submodular function $f$ over a set of elements,
  and we want to maximize its value subject to certain constraints. 
  Good approximation algorithms are known for such problems under  both monotone and 
  non-monotone submodular functions. We consider these problems in a  stochastic setting, where elements are not all active  and we
  can only get value from active elements. Each element $e$ is active independently
  with some known probability $p_e$, but we don't know the element's status 
   \emph{a priori}. We find it out only when we \emph{probe} the
  element $e$---probing reveals  whether it's active or not,
  whereafter we can use this information to decide which other elements
  to probe. Eventually, if we have a probed set $S$ and a subset
  $\text{active}(S)$  of active elements in $S$, we can
  pick any $T \sse \text{active}(S)$ and get value $f(T)$. 
Moreover, the sequence of elements we
  probe must satisfy a given \emph{prefix-closed constraint}---e.g.,
  these may be given by a matroid, or an orienteering constraint, or
  deadline, or precedence constraint, or an arbitrary downward-closed 
  constraint---if we can probe some sequence of
  elements we can probe any prefix of it.  What is a
  good strategy to probe elements to maximize the expected value?  

  In this paper we study the gap between adaptive and non-adaptive
  strategies for $f$ being a submodular or a fractionally subadditive
  (XOS) function. If this gap is small, we can focus on finding good
  non-adaptive strategies instead, which are easier to find as well as
  to represent.    We show that the adaptivity gap is  a constant for  monotone and non-monotone submodular functions, and logarithmic for XOS   functions of small \emph{width}. These bounds are nearly tight. Our techniques show new ways of
  arguing about the optimal adaptive decision tree for stochastic
  problems.
  \end{abstract}

%%%%%%%%%%%%%%%%%%%%%%%%%%%%%%%%%%%%%%%%%%%%

\section{Introduction}
\label{sec:introduction}  

Consider the problem of maximizing a submodular function $f$ over a set
of elements, subject to some given constraints. This has been a very
useful abstraction for many problems, both theoretical (e.g., the
classical $k$-coverage problem~\cite{SW10}), or practical (e.g., the
influence maximization problem~\cite{KKT15}, or many problems in machine
learning~\cite{Krause13}). We now know how to perform constrained
submodular maximization, both when the function is
monotone~\cite{FNW78,CCPV11}, and  also when the function may be
non-monotone but
non-negative~\cite{FMV-SICOMP11,LMNS-STOC09,FNS-FOCS11}.  In this paper
we investigate how well we can solve  this problem if the instance is not
deterministically known up-front, but there is uncertainty in the input.

Consider the following setting. We  have a submodular function over
a ground set of elements $U$. But the elements are not all active, and
we can get value only for active elements.  The bad news is that \emph{a
  priori} we don't know the elements' status---whether it is active or
not. The good news is that each element $e$ is active independently with
some known probability $p_e$.  We find out an element $e$'s status only
by \emph{probing} it. Once we know its status, we can use this
information to decide which other elements to probe next, and in what
order. (I.e.,  be \emph{adaptive}.) We have some constraints on
which subsets we are allowed to probe. Eventually, we stop with some
probed set $S$ and  a known subset $\text{active}(S)$ of the active
elements in $S$. At that time we can pick any $T \sse \text{active}(S)$
and get value $f(T)$.\footnote{If the function is monotone, clearly we
  should choose $T = \text{active}(S)$.}  What is a good strategy to
probe elements to maximize  expected value?

Since this sounds quite abstract, here is an example. In the setting of
influence maximization, the ground set is a set of email addresses (or
Facebook accounts), and for a set $S$ of email addresses $f(S)$ is the
fraction of the network that can be influenced by seeding the set $S$.
But not all email addresses are still active. For each email address
$e$, we know the probability $p_e$ that it is active. (Based, e.g., on
when the last time we know it was used, or some other machine learning
technique.) Now due to time constraints, or our anti-spam policies, or
the fact that we are risk-averse and do not want to make introductory
offers to too many people---we can only probe some $K$ of these
addresses, and make offers to the active ones in these $K$, to maximize
our expected influence. Observe that it makes sense to be adaptive---if
\texttt{t.theorist@cs.cmu.edu} happens to be active we may not want to
probe \texttt{t.theorist@cmu.edu}, since we may believe they are the
same person.

A different example is  robot path-planning. We have a robot that
can travel at most distance $D$ each day and is trying the maximize
value by picking items. The elements are locations, and an
element is active if the location has an item to be picked up. (Location
$e$ has an item with probability $p_e$ independent of all others.)
Having probed $S$, if $T$ is the subset of active locations, the value
$f(T)$ is some submodular function of the set of elements---e.g., the
number of distinct items. 

There are other examples: e.g., 
Bayesian mechanism design problems  (see~\cite{GN13} for details) and stochastic set cover
problems that arise in database applications~\cite{LiuPRY08,DHK14}.

The question that is of primary interest to us is the following:
\emph{Even though our model allows for adaptive queries, what is the
  benefit of this adaptivity?} Note that there is price to adaptivity:
the optimal adaptive strategy may be an exponentially-large decision
tree that is difficult to store, and also may be computationally
intractable to find. Moreover, in some cases the adaptive strategy would
require us to be sequential (probe one email address, then probe the
next, and so on), whereas a non-adaptive strategy may be just a set of
$K$ addresses that we can probe in parallel. So we want to bound the
\emph{adaptivity gap}: the ratio between the expected value of the best
adaptive strategy and that of the best non-adaptive strategy. 
Secondly, if this adaptivity gap is small, we would like to find the
best non-adaptive strategy efficiently (in polynomial time). This would
give us our ideal result: a non-adaptive strategy that is within a small
factor of the best adaptive strategy.

The goal of this work is to get such results for as broad a class of
functions, and as broad a class of probing constraints as possible.
 Recall that we were not allowed to probe all the elements, but only those which
satisfied some problem-specific constraints (e.g., probe at
most $K$ email addresses, or probe a set of locations that can be
reached using a path of length at most $D$.)

\subsection{Our Results}
\label{sec:results}

In this paper, we allow very general probing constraints: the sequence
of elements we probe must satisfy a given \emph{prefix-closed
  constraint}---e.g., these may be given by a matroid, or an
orienteering constraint, or deadline, or precedence constraint, or an
arbitrary downward-closed constraint---if we can probe some sequence of
elements we can probe any prefix of it. We cannot hope to get small
adaptivity gaps for arbitrary functions (see \S\ref{sec:adapgapbad} 
for a monotone $0-1$ function  where the gap is exponential in $n$ even for
 cardinality constraint),
 and hence we have to look at
interesting sub-classes of functions.

\paragraph{Submodular Functions.} Our first set of results are for the
case where the function $f$ is a non-negative submodular function. The
first result is for monotone functions.

\begin{theorem}[Monotone Submodular]
  \label{thm:monotSubmod} 
  For any monotone non-negative submodular function $f$ and any
  prefix-closed probing constraints, the stochastic probing problem has
  adaptivity gap at most $3$.
\end{theorem}

The previous results in this vein either severely restricted the
function type (e.g., we knew a logarithmic gap for matroid rank
functions~\cite{GNS16}) or the probing constraints (e.g., Asadpour et
al.~\cite{AN16} give a gap of $\frac{e}{e-1}$ for matroid probing
constraints). We discuss these and other prior works in
\S\ref{sec:related}.

There is a lower bound of $\frac{e}{e-1}$ on the adaptivity gap for
monotone submodular functions with prefix-closed probing constraints (in
fact for the rank function of a partition matroid, with the constraint
being a simple cardinality constraint). It remains an interesting open problem 
 to close this gap.

We then turn to non-monotone submodular functions, and again give a
constant adaptivity gap. While the constant can be improved slightly, we
have not tried to optimize it, focusing instead on clarity of
exposition.

\begin{theorem}[Non-Monotone Submodular]
  \label{thm:NmonotSubmod} 
  For any non-negative submodular function $f$ and any prefix-closed
  probing constraints, the stochastic probing problem has adaptivity gap
  at most $40$.
\end{theorem}

Both Theorems~\ref{thm:monotSubmod} and~\ref{thm:NmonotSubmod} just
consider the adaptivity gap. What about the computational question of
finding the best non-adaptive strategy? This is where the complexity of
the prefix-closed constraints come in. The problem of finding the best
non-adaptive strategy with respect to some prefix-closed probing
constraint can be reduced to the problem of maximizing a submodular
function with respect to the same  constraints.

\paragraph{XOS Functions.} We next consider more general classes of
functions. We conjecture that the adaptivity gap for all subadditive
functions is poly-logarithmic in the size of the ground set. Since we
know that any subadditive function can be approximated to within a
logarithmic factor by an XOS (a.k.a.\ \emph{max-of-sums}, or \emph{fractionally
  subadditive}) function~\cite{Dobzinski07}, and every XOS function is
subadditive, it suffices to focus on XOS functions. As a step towards
our conjecture, we show a nearly-tight logarithmic adaptivity gap for
monotone XOS functions of small ``\emph{width}'', which we explain below.

A monotone XOS function $f: 2^{[n]} \to \R_{\geq 0}$ is one that can
be written as the maximum of linear functions: i.e., there are vectors
$\C_i \in \R_{\geq 0}^n$ for some $i = 1 \ldots \W$ such that
\[ \textstyle{
f(S) := \max_i \big( \C_i^\intercal \chi_S \big) = \max_i
\big( \sum_{j \in S} \C_{i}(j) \big).  
}\] We define the \emph{width} of
(the representation of) an XOS function as $\W$, the number of linear
functions in this representation.  E.g., a width-$1$ XOS function is just
a linear function. In general, even representing submodular functions in
this XOS form requires an exponential width~\cite{BalcanHarvey-STOC11,BDFKNR-SODA12}. 

\begin{theorem}[XOS Functions]
  For any monotone XOS function $f$ of width $\W$, and any prefix-closed probing
  constraints, the adaptivity gap is $O(\log \W)$. Moreover, there are
  instances with $\W = \Theta(n)$ where the adaptivity gap is
  $\Omega\big(\frac{\log \W}{\log \log \W}\big)$.
\end{theorem}

In this case, we can also reduce the computation problem to
\emph{linear} maximization over the constraints.

\begin{theorem} 
  Suppose we are given a width-$\W$ monotone XOS function $f$ explicitly in the
  max-of-sums representation, and an oracle to maximize positive linear
  functions over some prefix-closed constraint. Then there exists an
  algorithm that runs in time $\poly(n,W)$ and outputs a
  non-adaptive strategy that has expected value at least an $\Omega(\log
  \W)$-fraction of the optimal adaptive strategy.
\end{theorem}

\subsection{Our Techniques}
\label{sec:techniques}

Before talking about our techniques, a word about previous approaches to
bounding the adaptivity gap. Several works, starting with the work of
Dean et al.~\cite{DGV04} have used geometric
``relaxations'' (e.g., a linear program for linear
functions~\cite{DGV04}, or the multilinear extension for submodular
settings~\cite{ANS,ASW14}) to get an estimate of the value achieved by
the optimal adaptive strategy. Then one  tries to find a non-adaptive
strategy whose expected value is not much less than this relaxation.
This is particularly successful when the probing constraints are
amenable to being captured by linear programs---e.g., matroid or
knapsack constraints. Dealing with general constraints (which include
orienteering constraints, where no good linear relaxations are
known) means we cannot use this approach. 

The other approach is to argue about the optimal decision-tree directly.
An induction on the tree was used, e.g., by Chen et al.~\cite{CIKMR09}
and Adamczyk~\cite{A11} to study stochastic matchings. A different
approach is to use concentration bounds like Freedman's inequality to
show that for most paths down the tree, the function on the path behaves
like the path-mean---this was useful in \cite{GKNR12-soda}, and also in
our previous work on adaptivity gaps of matroid rank
functions~\cite{GNS16}. However, this approach seems best suited to
linear functions, and  loses logarithmic factors due to the need for
union bounds.

Given that we  prove  a general result for any submodular function,
 how do we show a good
non-adaptive strategy? Our approach is to take a random path down the
tree (the randomness coming from the element activation probabilities)
and to show the expected value of this path, when viewed as a
non-adaptive strategy, to be good. To prove this, surprisingly, we use
induction. It is surprising because the natural induction down the tree
does not seem to work. So we perform a non-standard inductive argument,
where we consider the all-\no path (which we call the \emph{stem}), show
that a non-adaptive strategy would get value comparable to the decision
tree on the stem, and then induct on the subtrees hanging off this
stem. The proof for monotone functions, though basic, is
subtle---requiring us to change representations and view things
``right''. This appears in \S\ref{sec:monotone}.

For non-monotone submodular functions, the matter is complicated by the
fact that we cannot pick elements in an ``online'' fashion when going
down the tree---greedy-like strategies are bad for non-monotone
functions. Hence we pick elements only with some probability, and show
this gives us a near-optimal solution. The argument is complicated by
the fact that having picked some elements $X$, the marginal-value
function $f_X(S) := f(X\cup S) - f(X)$ may no longer be non-negative.

Finally, for monotone XOS functions of small width, we use the approach based on
Freedman's concentration inequality to show that a simple algorithm that
either picks the set optimizing one of the linear functions, or a single
element, is within an $O(\log \W)$ factor of the optimum. We then give a
lower bound example showing an (almost-)logarithmic factor is necessary,
at least for $\W = O(n)$.

\subsection{Related Work}
\label{sec:related}

The adaptivity gap of stochastic packing problems has seen much
interest: e.g., for knapsack~\cite{DGV04,BGK11,M14}, packing integer
programs~\cite{DGV05,CIKMR09,BGLMNR12}, budgeted multi-armed
bandits~\cite{GM07,GKMR11,LiY13,M14} and
orienteering~\cite{GuhaM09,GKNR12-soda,BN14}. All except the orienteering
results rely on having relaxations that capture the constraints of the
problem via linear constraints.

For stochastic monotone submodular functions where the probing
constraints are given by matroids, Asadpour et al.~\cite{AN16} bounded
the adaptivity gap by $\frac{e}{e-1}$; Hellerstein et al.~\cite{HKL15}
bound it by $\frac1\tau$, where $\tau$ is the smallest probability of
some set being materialized. (See also~\cite{LiuPRY08,DHK14}.)

The work of Chen et al.~\cite{CIKMR09} (see
also~\cite{A11,BGLMNR12,BCNSX15,AGM15}) sought to maximize the size of a
matching subject to $b$-matching constraints; this was motivated by
applications to online dating and kidney exchange. More generally, see,
e.g.~\cite{RSU05,AR12}, for pointers to other work on kidney exchange
problems. The work of~\cite{GN13} abstracted out the general problem of
maximizing a function (in their case, the rank function of the
intersection of matroids or knapsacks) subject to probing constraints
(again, intersection of matroids and knapsacks). This was
improved and generalized by Adamczyk, et al.~\cite{ASW14} to submodular
objectives. All these results use LP relaxations, or non-linear
geometric relaxations for the submodular settings.

The previous work of the authors~\cite{GNS16} gave results for the case
where $f$ was the rank function of matroids (or their intersections).
That work bounded the adaptivity gap by logarithmic factors, and gave
better results for special cases like uniform and partition matroids. 
This work both improves the quantitative bounds (down to small
constants), generalizes it to all submodular functions with the hope of
getting to all subadditive functions, and arguably also makes the proof
simpler.

\section{Preliminaries and Notation}
\label{sec:prelims}

We denote the ground set by $X$, with $n = |X|$. Each element $e \in X$
has an associated probability $p_e$. Given a subset $S \sse X$ and
vector $\p = (p_1, p_2, \ldots, p_n)$, let $S(\p)$ denote the distribution
over subsets of $S$ obtained by picking each element $e \in S$
independently with probability $p_e$. (Specifying a single number $p \in
[0,1]$ in $S(p)$ indicates each element is chosen with probability $p_e
= p$.)

A function $f: 2^X \to \R$ is 
\begin{itemize}[topsep=0pt,itemsep=0.1em]
\item \emph{monotone} if $f(S) \leq f(T)$ for all $S \sse T$.

\item \emph{linear} if there exist $a_i \in \R$ for each $i \in X$ such
  that $f(S) = \sum_{i \in S} a_i$.

\item \emph{submodular} if $f(A \cup B) + f(A \cap B) \leq f(A) + f(B)$
  for all $A, B \sse X$. We will normally assume that $f$ is
  non-negative and $f(\emptyset) = 0$.

\item \emph{subadditive} if $f(A \cup B) \leq f(A) + f(B)$. A
  non-negative submodular function is clearly subadditive.

\item \emph{fractionally subadditive} (or \emph{XOS}) if $f(T) \leq
  \sum_i \alpha_i f(S_i)$ for all $\alpha_i \geq 0$ and $\chi_T = \sum_i
  \alpha_i \chi_{S_i}$. \footnote{Our definitions of fractionally
    subadditive/XOS differ slightly from those in the literature, since
    we allow non-monotonicity in our functions. See \S\ref{sec:non-monotone-xos} for
    a discussion.}
 
  An alternate characterization: a function is XOS if there exist linear
  functions $\C_1, \C_2, \ldots, \C_w: 2^X \to \R$ such that $f(X) =
  \max_j \{ \C_j(X) \}$. The \emph{width} of an XOS function is the
  smallest number $\W$ such that $f$ can be written as the maximum over
  $\W$ linear functions.
\end{itemize}

All objective functions $f$ that we deal with are non-negative with $f(\emptyset)=0$. 

Given any function $f: 2^X \to \R$, define $\fmax(S) := \max_{T \sse
  S} f(T)$ to be the maximum value subset contained within $S$. The
function $f$ is monotone if and only if $\fmax = f$. In general, $\fmax$
may be difficult to compute given access to $f$.
However, Feige et al.~\cite{FMV-SICOMP11} show that for submodular functions 
\begin{gather}
\textstyle \frac14 \fmax(S) \leq \E_{R \sim S(\frac12)}[f(R)] \leq
\fmax(S). \label{eq:fmv}
\end{gather}

Also, for a subset $S$, define the ``contracted'' function 
$f_S(T) := f(S \cup T) - f(S)$. Note that if $f$ is non-monotone, then $f_S$ may be
negative-valued even if $f$ is not.

\paragraph{Adaptive Strategies}
An adaptive strategy tree $\T$ is a a binary rooted tree where every
internal node $v$ represents some element $e \in X$ (denoted by $\elt(v)
= e$), and has two outgoing arcs---the \yes arc indicating the node
to go to if the element $e = \elt(v)$ is active (which happens with
probability $p_e$) when probed, and the \no arc indicating the node
to go to if $e$ is not active (which happens with the remaining
probability $q_e = 1-p_e$). No element can be represented by two
different nodes on any root-leaf path. Moreover, any root-leaf path in
$T$ should be feasible according to the  constraints. Hence, each
leaf $\ell$ in the tree $\T$ is associated with the root-path $P_\ell$:
the elements probed on this path are denoted by $\elt(P_\ell)$. Let
$A_\ell$ denote the active elements on this path $P_\ell$---i.e., the
elements represented by the nodes on $P_\ell$ for which we took the
\yes arc.

The tree $\T$ naturally gives us a probability distribution $\pi_\T$
over its leaves: start at the root, at each node $v$, follow the \yes
branch with probability $p_{\elt(v)}$ and the ``no'' branch otherwise,
to end at a leaf. 

Given a submodular function $f$ and a tree $\T$, the associated adaptive
strategy is to probe elements until we reach a leaf $\ell$, and then to
pick the max-value subset of the active elements on this path $P_\ell$. 
Let $\adap(\T, f)$ denote the expected value
obtained this way; it can be written compactly as
\begin{gather}
  \adap(\T, f) := \E_{\ell \gets \pi_\T} [\fmax(A_\ell)]. \label{eq:adap}
\end{gather}

\begin{definition}[stem of $\T$]
  \label{defn:stem} For any adaptive strategy tree $\T$ the \emph{stem}
  represents the all-\no path in $\T$ starting at the root, i.e., when
  all the probed elements turn out inactive.
\end{definition}

\begin{definition}[\depth of $\T$] 
  \label{defn:depth} The \emph{\depth} of a strategy tree $\T$ is the
  maximum number of \emph{active} nodes that \adap sees along a root-leaf path
  of $\T$.
\end{definition}

Note that this notion of \depth is not the same as that of depth used
for trees: it measures the number of \yes-arcs on the path from the root
to the leaf, rather than just the number of arcs seen on the path. This
definition is inspired by the induction we will do in the submodular
sections.

We can also define the \emph{natural non-adaptive} algorithm given the
tree $\T$: just pick a leaf $\ell \gets \pi_\T$ from the distribution
given by $\T$, probe all elements on that path, and choose the max-value
subset of the active elements. We denote the
expected value by $\alg(\T, f)$:
\begin{gather}
  \alg(\T, f) := \E_{\ell \gets \pi_T} [ \E_{R \sim X(p)} [ \fmax(R \cap
  \elt(P_\ell)) ] ]. \label{eq:alg}
\end{gather}

\newcommand{\agleft}{\big}
\newcommand{\agright}{\big}
\newcommand{\ts}{\textstyle}

\section{Monotone Non-Negative Submodular Functions}
\label{sec:monotone}

We now prove Theorem~\ref{thm:monotSubmod}, and bound the adaptivity gap
for \emph{monotone} submodular functions $f$ over any prefix-closed set
of constraints. The idea is a natural one in retrospect: we take an
adaptive tree $\T$, and show that the natural non-adaptive strategy
(given by choosing a random root-leaf path down the tree, and probing
the elements on that path) is within a factor of $3$ of the adaptive
tree. The proof is non-trivial, though. One strategy is to induct on the
two children of the root (which, say, probes element $e$), but note that
the adaptive and non-adaptive algorithms recurse having seen different
sets of active elements.\footnote{Adaptive sees $e$ as active when
  it takes the \yes branch (with probability $p_e$), and nothing as
  active when
  taking the \no branch. Non-adaptive recurses on the \yes branch with
  the same probability $p_e$ because it picks a random path down the
  tree---but it then also probes $e$. So when it takes the \yes branch,
  it has either seen $e$ as active (with probability $p_e$) or not (with
  probability $1-p_e$). Hence the set of active elements on both sides
  are quite different.} This forced previous results to proceed along
different lines, using massive union bounds over the paths in the
decision tree, and hence losing logarithmic factors. They were also
restricted to matroid rank functions, instead of all submodular
functions.

A crucial insight in our proof is to focus on the \emph{stem} of the
tree (the all-\no path off the root, see Definition~\ref{defn:stem}),
and induct on the subtrees hanging off this stem. Again we have issues
of adaptive and non-adaptive recursing with different active elements,
but we control this by giving the adaptive strategy some elements for
free, and contracting some elements in the non-adaptive strategy without
collecting value for them. The proof for non-monotone functions in
\S\ref{sec:non-mono} will be even more tricky, and will build on ideas
from this monotone case. Formally, the main technical result is the
following:
\begin{theorem} 
  \label{thm:monadapvsalg} 
  For any adaptive strategy tree $\T$, and any monotone non-negative
  submodular function $f: 2^X \to \R_{\geq 0}$ with $f(\emptyset) = 0$,
  \[ \alg(\T, f) \geq \nicefrac13~\adap(\T, f). \]
\end{theorem}
Theorem~\ref{thm:monotSubmod} follows by the observation that each
root-leaf path in $\T$ satisfies the prefix-closed constraints, which
gives us a feasible non-adaptive strategy.
Some comments on the proof: because the function $f$ is monotone, $\fmax
= f$. Plugging this into~(\ref{eq:adap}) and~(\ref{eq:alg}), we want to
show that
\begin{gather}
\E_{\ell \gets \pi_T} [ \E_{R \sim X(p)} [ f(R \cap
  \elt(P_\ell)) ] ]
~~\geq~~ \nicefrac13\,   \E_{\ell \gets \pi_\T} [f(A_\ell)]. \label{eq:simpler-mono}
\end{gather}
Since both expressions take expectations over the random path, the proof
proceeds by induction on the \depth of the tree. (Recall the
 definition of \depth in Definition~\ref{defn:depth}.) We
argue that for the stem starting at the root, \alg gets a value close to
\adap in expectation (Lemma~\ref{lemma:stemmass}). However, to induct on the subtree that the
algorithms leave the stem on, the problem is that the two algorithms may
have picked up different active elements on the stem, and hence the
``contracted'' functions may look very different. The idea now is to
give \adap the elements picked by \alg for ``free'' and disallow \alg
(just for the analysis) to pick elements picked by \adap after exiting
the stem. Now both the algorithms work after contracting the same set of
elements in $f$, and we are able to proceed with the induction.

\subsection{Proof of Theorem~\ref{thm:monadapvsalg}}

\begin{proof}
  We prove by induction on the \depth of the adaptive strategy tree
  $\T$. For the base case of \depth~$0$, $\T$ contains exactly one node,
  and there are no internal nodes representing element. Hence both
  $\alg$ and $\adap$ get zero value, so the theorem is vacuously true.

  To prove the induction step, recall that the stem is the path in $\T$
  obtained by starting at the root node and following the \texttt{no}
  arcs until we reach a leaf. (See Figure~\ref{fig:stem}.)  Let $v_1,v_2, \ldots, v_\ell$ denote the nodes along the
  stem of $\T$ with $v_1$ being the root and $v_\ell$ being a leaf; let
  $e_i = \elt(v_i)$.  For $i\geq 1$, let $\T_i$ denote the subtree
  hanging off the \texttt{yes} arc leaving $v_i$.  The probability that
  a path following the probability distribution $\pi_\T$ enters $\T_i$
  is $p_i ~\prod_{j<i}q_j$, where $p_i = 1-q_i$ denotes the probability
  that the $i^{th}$ element is active. 
\tikzstyle{point}=[circle, draw, fill=black!30, inner sep=0pt, minimum width=2pt]

\tikzstyle{block}=[draw opacity=0.7,line width=1.4cm]
\tikzstyle{graphnode}=[circle, draw, fill=black!15, inner sep=0pt, minimum width=12pt]

\tikzstyle{input}=[rectangle, draw, fill=black!75,inner sep=3pt, inner ysep=3pt, minimum width=4pt]
\tikzstyle{unmatched}=[graphnode,fill=black!0]
\tikzstyle{shaded}=[graphnode,fill=black!20]
\tikzstyle{matched}=[graphnode,fill=black!100]  	
\tikzstyle{matching} = [ultra thick]
\tikzset{
    >=stealth',
    pil/.style={
           ->,
           thick,
           shorten <=2pt,
           shorten >=2pt,}
}
\tikzset{->-/.style={decoration={
  markings,
  mark=at position .5 with {\arrow{>}}},postaction={decorate}}}

\begin{figure}[ht]
\begin{center}
\begin{tikzpicture}[thin,scale=1]

	\foreach \y in {0,1,2,3}{
	\draw [line width=1mm] (-\y,-\y) -- (-\y-1,-\y-1);
	\draw [thick] (-\y,-\y) -- (-\y+1,-\y-1);
	}

	\draw [->-,thick,color=blue] (0,0) to [out=200,in=90] (-1,-1);
	\draw [->-,thick,color=blue] (-1,-1) to [out=290,in=30] (-2,-2);
	\draw [->-,thick,color=blue] (-2,-2) to [out=200,in=90] (-3,-3);
	\draw [->-,thick,color=blue] (-3,-3) to [out=240,in=90] (-2,-4);
	
	\foreach \y in {0,1,2,3}{
	\node at (-\y,-\y) [graphnode]{};}

	\foreach \y in {1,2,3,4}{
	\draw [thick,fill=black!15] (2-\y,-\y) -- (2-\y-0.5, -\y-1) -- (2-\y+0.5, -\y-1)-- (2-\y,-\y);
	\node at (2-\y,-\y-1)[label=above:$\T_{\y}$] {};}
	
	\node at (2.5,-1.5)[label=above:\yes] {};
	\node at (-3,-1.5)[label=above:\no] {};

\end{tikzpicture}
\end{center}
\caption{Adaptive strategy tree $\T$. The thick line shows the all-\no path. The arrows show the path taken by \adap. In this example $i=4$ and $S_i = \{e_1,e_2,e_3,e_4 \}$.}
\label{fig:stem}
\end{figure}

  Let $S_i = \{e_1, e_2, \ldots, e_i\}$ be the first $i$ elements probed
  on the stem, and $R_i \sim S_i(\p)$ be a random subset of $S_i$ that
  contains each element $e$ of $S_i$ independently w.p.\ $p_e$.  We can
  now rewrite \adap and \alg in a form more convenient for
  induction. Here we recall the definition of a marginal with respect to
  subset $Y$: $f_Y(S) := f(Y
  \cup S) - f(Y)$. Note that the leaf $v_\ell$ has no associated
  element; to avoid special cases we define a dummy element $e_\ell$
  with $f(\{e_\ell\}) = 0$ and $f_{\{e_{\ell}\}} = f$.

  \begin{claim} \label{claim:montboundalgs} Let $I$ be the r.v.\
    denoting the index of the node at which a random walk according to
    $\pi_\T$ leaves the stem. (If $I = \ell$ then the walk does not
    leave the stem, and $\T_\ell$ is a \depth-zero tree.) Then,
    \begin{align}
      \adap(\T, f)  &= \E_I \agleft[ f(e_I) + \adap(\T_I, f_{\{e_I\}})
      \agright] \label{eq:monsub1} \\
      &\leq  \E_{I,R \sim S_I(\p)} \agleft[ f(e_I) + f(R) +
        \adap(\T_I, f_{R \cup e_I}) \agright] \label{eq:monsub1a}\\
      \alg_H(\T) &= \E_{I, R \sim S_I(\p)}\agleft[ f(R) + \alg(\T_I, f_{R})\agright]
      \label{eq:monsub2}\\
      &\geq \E_{I,R \sim S_I(\p)}\agleft[ f(R) + \alg(\T_I, f_{ R \cup
          e_I})\agright] \label{eq:monsub2a}
    \end{align}
  \end{claim}
  \begin{proof}
    Equation~(\ref{eq:monsub1}) follows from the definition of \adap;
    \eqref{eq:monsub1a} follows from the monotonicity of $f$. (We are
    giving the adaptive strategy elements in $R$ ``for free''.)
    Equation~(\ref{eq:monsub2}) follows from the definition of \alg, and
    (\ref{eq:monsub2a}) uses the consequence of submodularity that
    marginals can only decrease for larger sets.
  \end{proof}

  Observe the expressions in~(\ref{eq:monsub1a}) and~(\ref{eq:monsub2a})
  are ideally suited to induction. Indeed, since the function $f_{R \cup
    e_I}$ also satisfies the assumptions of
  Theorem~\ref{thm:monadapvsalg}, and the height of $\T_i$ is smaller
  than that of $\T$, we use induction hypothesis on $\T_i$ with monotone
  non-negative submodular function $f_{R\cup E_I}$ to get
  \begin{gather*}
   \E_{I,R \sim S_I(\p)}\agleft[ \alg(\T_I, f_{ R \cup e_I})\agright] \geq
   \nicefrac13 \,
   \E_{I,R \sim S_I(\p)}\agleft[ \adap(\T_I, f_{ R \cup e_I})\agright].
  \end{gather*}
  Finally, we use the following Lemma~\ref{lemma:stemmass} to show that
  \begin{gather*}
   \E_{I,R \sim S_I(\p)}\agleft[ f(R)  \agright] \geq
   \nicefrac13 \,
   \E_{I,R \sim S_I(\p)}\agleft[ f(R) + f(e_I) \agright].
  \end{gather*}
  Substituting these two into~(\ref{eq:monsub1a})
  and~(\ref{eq:monsub2a}) finishes the induction step.
\end{proof}

\begin{lemma} 
  \label{lemma:stemmass} 
  Let $I$ be the r.v.\ denoting the index of the node at which a random
  walk according to $\pi_\T$ leaves the stem. (If $I = \ell$ then the
  walk does not leave the stem.)
  Then,
  \begin{gather*}
   \E_{I,R \sim S_I(\p)}\agleft[ f(R)  \agright] \geq
   \nicefrac12 \,
   \E_{I}\agleft[ f(e_I) \agright].
  \end{gather*}
\end{lemma}
\begin{proof}
  For brevity, we use $\E_{I,R}[\cdot]$ as shorthand for $\E_{I,R \sim
    S_I(\p)}[\cdot]$ in the rest of the proof.  We prove this lemma by
  showing that
  \begin{gather}
   \E_{I,R}\agleft[ f(R)  \agright] \geq
   \E_{I,R}\agleft[ \max_{e_j \in R} f(e_j)  \agright] \geq
   \nicefrac12 \,
   \E_{I}\agleft[ f(e_I) \agright].
  \end{gather}
  The first inequality uses monotonicity. 
  The
  rest of the proof shows the latter inequality.

  For any real $x \geq 0$, let $W_x$ denote the indices of the elements
  $e_j$ on the stem with $f(e_j) \geq x$, and let $\overline{W}_x$
  denote the indices of stem elements not in $W_x$.  Then,
  \begin{align}
    \E_I\agleft[f(e_I) \agright] &= \int_0^{\infty} \Pr_I[f(e_I) \geq x] \,
    dx = \int_0^{\infty} \Pr_I[ I \in W_x] \, dx =  \int_0^{\infty}
    \sum_{i \in W_x} \agleft( p_i \,  \prod_{j<i}  q_j \agright) \, dx,
    \label{eq:monotadap}
  \end{align}
  where the last equality uses that the probability of exiting stem at
  $i$ is $ p_i \, \prod_{j<i} q_j $.

  On the other hand, we have 
{\small  \begin{align} & \E_{I,R}\agleft[ \max_{e_j \in R} f(e_j)
  \agright]  =  \int_0^{\infty} \Pr_{I,R}[\max_{e_j
      \in R} f(e_j) \geq x] \, dx   = \int_0^{\infty} \Pr_{I,R}[ R\cap W_x\ne \emptyset ] \, dx  \notag \\
       & = \int_0^{\infty}  \sum_{k \in W_x}
    \Pr_{I,R}[e_k\in R \mbox{ and $e_{j} \not\in R$ for all $j<k$ with $j\in W_x$}] \, dx  \notag \\
    &= \int_0^{\infty}  \sum_{k \in W_x} \Pr_I[ I \geq k] \cdot \Pr[e_k \text{ active}]  \cdot
    \Pr_I[\text{ $e_j$ inactive for all $j<k$ with $j\in W_x$ } ]
    \,dx \label{eq:nonadp-stem1}\\
&= \int_0^{\infty}  \sum_{k \in W_x}  \agleft( \prod_{j<k}  q_j  \agright) \cdot p_k \cdot \agleft( \prod_{j<k \, \&\,  j \in W_x} q_j \agright) \, dx \,=\, \int_0^{\infty}  \sum_{k \in W_x} \agleft( \prod_{j<k \, \&\, j \in W_x} q_j^2 \agright) \cdot \agleft(
      \prod_{j<k \, \&\, j \not\in W_x} q_j \agright) \cdot p_k \, dx. \label{eq:nonadp-stem2}
\end{align}}
Recall that $R\sim S_I(\p)$. Above \eqref{eq:nonadp-stem1} is because, for $e_k$ to be the first element in $W_x\cap R$ (i) the index $I$ must go past $k$, (ii) $e_k$ must be 
      active, and  (iii) 
all elements before $k$ on the
stem with indices in $W_x$  must be inactive (which are all independent events). Equation \eqref{eq:nonadp-stem2} is by definition of these probabilities. 
Renaming $k$ to $i$,
\begin{align}\label{eq:monotalg}
\E_{I,R}\agleft[ \max_{e_j \in R} f(e_j)  \agright] =  \int_0^{\infty}  \sum_{i \in W_x} \agleft( p_i \agleft(\prod_{j<i \, \&\,  j\in W_x} q_j^2 \agright) \agleft( \prod_{j<i \, \&\,  j\not\in W_x} q_j \agright) \agright) \, dx.
\end{align}

To complete the proof, we compare equations~(\ref{eq:monotadap})
and~(\ref{eq:monotalg}) and want to show that for every $x$,
\begin{align} \label{eq:finalcomp} \sum_{i \in W_x} \agleft( p_i
    \agleft(\prod_{j<i \, \&\, j\in W_x} q_j^2 \agright) \agleft( \prod_{j<i
        \, \&\, j \not\in {W}_x} q_j \agright) \agright) \geq \nicefrac12
  \sum_{i\in W_x} \agleft( p_i \prod_{j<i} q_j \agright).
 \end{align}
While the expressions look complicated, things simplify considerably
when we condition on the outcomes of elements outside $W_x$. Indeed, 
observe that the LHS of~(\ref{eq:finalcomp}) equals 
\begin{align}
  &\phantom{=~~}  \E_{\overline{W}_x} \agleft[\sum_{i \in W_x} \agleft( p_i
      \agleft(\prod_{j<i \, \&\, j\in W_x} q_j^2 \agright) \agleft( \prod_{j<i
          \, \&\, j\not\in W_x} \textbf{1}_{q_j} \agright) \agright)
  \agright],\\
\intertext{where $\textbf{1}_{q_j}$ is an independent indicator r.v.\ taking
  value $1$ w.p.\ $q_j$, and we take the expectation over coin tosses
  for elements of stem outside $W_x$. Similarly, the RHS
  of~(\ref{eq:finalcomp}) is}
\nicefrac12\, \sum_{i\in W_x} \agleft( p_i \prod_{j<i} q_j \agright) &=
  \E_{\overline{W}_x} \agleft[ \nicefrac12\, \sum_{i \in W_x} \agleft( p_i
      \agleft(\prod_{j<i \, \&\, j\in W_x} q_j \agright) \agleft( \prod_{j<i
          \, \&\, j\not\in W_x} \textbf{1}_{q_j} \agright) \agright)
  \agright].
\end{align}
Hence, after we condition on the elements outside $W_x$, the remaining
expressions can be related using the following claim.  
\begin{restatable}{claim}{stemineq} 
  \label{claim:cuteineq} For any ordered set $A$ of
  probabilities $ \{a_1, a_2, \ldots, a_{|A|} \} $, let $b_j$ denote $
  1- a_j$ for $ j \in [1, |A|]$. Then,
  \[\sum_{i} a_i \agleft(\prod_{j<i } b_j \agright)^2  \geq \nicefrac12 \sum_{i}
  a_i  \agleft(\prod_{j<i } b_j \agright) \] 
\end{restatable}
\begin{proof}
\begin{align*}
\sum_i a_i\bigg(\prod_{j<i} b_j \bigg)^2  &=  \sum_i \frac{1-b_i^2}{1+b_i}\bigg(\prod_{j<i} b_j \bigg)^2  \geq \frac12 \sum_i (1-b_i^2)\bigg(\prod_{j<i} b_j^2\bigg) \\
& =^{(\star)} \frac12 \agleft( 1- \prod_{i} b_i^2 \agright)  = \frac12 \agleft( 1- \prod_{i} b_i \agright)  \agleft( 1+ \prod_{i} b_i \agright) \\
& \geq \frac12  \agleft( 1- \prod_{i} b_i \agright) =^{(\star)}
\frac12 \sum_{i}
  a_i  \agleft(\prod_{j<i } b_j \agright),
\end{align*}
where we have repeatedly used $a_j + b_j = 1$ for all $j$. The
equalities marked $(\star)$ move between two ways of expressing the
probability of at least one ``heads'' when the tails probability is
$b_j^2$ and $b_j$ respectively.
\end{proof}

Applying the claim to the elements in $W_x$, in order of their
distance from the root, completes the proof.
\end{proof}

\subsection{Lower Bounds}
\label{sec:submod-lower-bd}

Our analysis cannot be substantially improved, since
Claim~\ref{claim:cuteineq} is tight. Consider the setting with $|A|$
being infinite for now, and $a_i = \varepsilon$ for all $i$. Then the
LHS of Claim~\ref{claim:cuteineq} is $\varepsilon \sum_i (1 -
\varepsilon)^{2(i-1)} = \frac{\varepsilon}{1 - (1 - \varepsilon)^2}
\approx \nicefrac12 + O(\varepsilon)$, whereas the sum on the right is $1$.
Making $|A|$ finite but large compared to $\nicefrac1\varepsilon$ would give
similar results.

However, there is still hope that a smaller adaptivity gap can be proved
using other techniques. The best lower bound on adaptivity gaps for
monotone submodular functions we currently know is $\frac{e}{e-1}$. The
function is the rank function of a partition matroid, where the universe
has $k$ parts (each with $k^2$ elements) for a total of $n = k^3$
elements. Each element has $p_e = \nicefrac1k$. The probing constraint is a
cardinality constraint that at most $k^2$ elements can be probed.  In
this case the optimal adaptive strategy can get $(1 - o(1))k$ value,
whereas any non-adaptive strategy will arbitrarily close to $(1 -
\nicefrac1e)k$ in expectation. (See, e.g.,~\cite[Section~3.1]{AN16}.)

\subsection{Finding Non-Adaptive Polices}

A non-adaptive policy is given by a fixed sequence $\sigma = \langle
e_1, e_2, \ldots, e_k\rangle$ of elements to probe (such that $\sigma$
satisfies the given prefix-closed probing constraint. If $A$ is the set
of active elements, then the value we get is $\E_{A \sim X(\p)}[\fmax(A
\cap \{e_1, \ldots, e_k\})] = \E_A[f(A \cap \{e_1, \ldots, e_k\})]$, the
inequality holding for monotone functions. If we define $g(S) := \E_{A
  \sim X(\p)}[ f(X \cap A)]$, $g$ is also a monotone submodular
function. Hence finding good non-adaptive policies for $f$ is just
optimizing the monotone submodular function $g$ over the allowed
sequences. E.g., for the probing constraint being a matroid constraint,
we can get a $\frac{e}{e-1}$-approximation~\cite{CCPV11}; for it being
an orienteering constraint we can get an $O(\log n)$-approximation in
quasi-polynomial time~\cite{CP05}.

For non-monotone functions (discussed in the next section), we can
approximate the $\fmax(S)$ function by $E_{R \sim X(\nicefrac12)}[f(S
\cap R)]$, and losing a factor of $4$, reduce finding good non-adaptive
strategies to (non-monotone) submodular optimization over the probing
constraints.

\newcommand{\adapbar}{\overline{\adap}}
\newcommand{\algbar}{\overline{\alg}}
\newcommand{\nil}{\bot}

\section{Non-Monotone Non-Negative Submodular Functions}
\label{sec:non-mono}

We now prove Theorem~\ref{thm:NmonotSubmod}.  The proof
for the monotone case used monotonicity in several places, but perhaps
the most important place was to claim that going down the tree, both
\adap and \alg could add all active elements to the set. This ``online''
feature seemed crucial to the proof. In contrast, when the adaptive
strategy \adap reaches a leaf in the non-monotone setting, it chooses
the best subset within the active elements; a similar choice is done by
the non-adaptive algorithm. This is why we have $\fmax(A_\ell)$
in~(\ref{eq:adap}) versus $f(A_\ell)$ in~(\ref{eq:simpler-mono}).

Fortunately, Feige el al.~\cite{FMV-SICOMP11} show that for non-negative
non-monotone submodular functions, the simple strategy of picking every
active element independently w.p.\ half gives us a near-optimal possible
subset. Losing a factor of four, this result allows us to analyze the
performance relative to an adaptive {\em online} algorithm $\adap_{on}$
which selects (with probability $\frac12$) each probed element that
happens to be active. 
The
rest of the proof is similar (at a high level) to the monotone case: to
relate $\adap_{on}$ and \alg we bound them using comparable terms
($\adapbar$ and $\algbar$ in Definition~\ref{defn:proxy}) and apply induction. Altogether we will obtain:
\[ \textstyle \alg \stackrel{(\text{Lemma}~\ref{lem:big-lemma}(ii))}{\ge} \algbar
\stackrel{(\text{Lemma}~\ref{lem:algbar-adapbar})}{\ge} \frac15\cdot
\adapbar \stackrel{(\text{Lemma}~\ref{lem:big-lemma}(i))}{\ge}  \frac{1}{10}\cdot \adap_{on} \stackrel{(\ref{eq:fmv})}{\ge}
\frac{1}{40}\cdot \adap. \]

% %------------------------

In the inductive proof, we will work with ``contracted'' submodular
functions $g$ obtained from $f$, which may take negative values but have
$g(\emptyset)=0$.  In order to deal with such issues, the induction here
is more complex than in the monotone case.

We first define the surrogates $\adapbar$ and $\algbar$ for \adap and
\alg recursively as follows.
\begin{definition}\label{defn:proxy}
  For any strategy tree $\T$ and submodular function $g$ with
  $g(\emptyset)=0$, let
  \begin{OneLiners}
  \item $I$ be the node at which a random walk according to $\pi_\T$
    exits the stem.
  \item $R\sim S_I(\p)$ 
    where $S_I$ denotes the elements on the stem until node $I$.
  \item $J = \arg\max \{g(e) \mid e\in R, g(e)>0\}$ w.p.\ $\frac{1}{2}$ and
    $J=\nil$ w.p.\ $\frac{1}{2}$.
  \end{OneLiners}
  Then we define:
  \begin{gather*}
    \adapbar(\T,g) := \E_{I,J}\left[ g(I) + g(J) + \adapbar(\T_I,
      g_{I \cup J})\right] \quad \mbox{and} \quad \algbar(\T,g) :=
    \E_{I,J}\left[ g(J) + \algbar(\T_I, g_{I\cup J})\right].
  \end{gather*}
\end{definition}
Above we account for the
non-monotonicity of the function, via this process of random sampling
used in the definition of $\adapbar$ and $\algbar$. One problem with
following the proof from \S\ref{sec:monotone} is that when we induct on
the ``contracted'' function $f_S$ for some set $S$, this function may
not be non-negative any more. Instead, our proof considers the entire
path down the tree and argues about it at one shot; to make the analysis
easier we imagine that the non-adaptive algorithm picks at most one item
from the stem, i.e., the one with the highest marginal value.

\begin{restatable}{lemma}{nonmonolemma}
  \label{lem:big-lemma}
  For any strategy tree $\T$, the following hold:
  \begin{enumerate}
  \item[(i)] For any non-negative submodular function $f$,
    $\adapbar(\T,f)\ge \frac12 \adap_{on}(\T,f)$.
  \item[(ii)] For any submodular function $g$, $\alg(\T,g)\ge
    \algbar(\T,g)$.
  \end{enumerate}
\end{restatable}
 We make use of the following property of submodular functions.
\ignore{
\begin{lemma}[\cite{FMV-SICOMP11}, Lemma~2.2]
  \label{lem:FMVlem22}
  For any non-negative submodular function $f$ on set $A$,
  \[ \E_{R \sim A(p)}[f(R) ] \geq (1-p) f(\emptyset) + p f(A).   \]
\end{lemma}
}
\begin{lemma}[\cite{BFNS-SODA14}, Lemma~2.2]\label{lem:BFNS}
  For any non-negative submodular function $h: 2^A \to \R_{\geq 0}$
  (possibly with $h(\emptyset) \neq 0$) let $S\sse A$ be a random subset
  that contains each element of $A$ with probability \emph{at most} $p$
  (and not necessarily independently). Then, $\E_{S} [f(S)] \geq (1-p)
  \cdot f(\emptyset)$.
\end{lemma}

\begin{proof}[Proof of Lemma~\ref{lem:big-lemma}]
  We condition on a random leaf $\ell$ drawn according to
  $\pi_{\T}$. Let $I_1,\ldots ,I_d$ denote the sequence of nodes that
  correspond to active elements on the path $P_\ell$, i.e., $I_1$ is the
  point where $P_\ell$ exits the stem of \T, $I_2$ is the point where
  $P_\ell$ exits the stem of $\T_{I_1}$ etc. Then, the adaptive online
  value is exactly $f(\{I_1,\ldots I_d)$. For any $k=1,\ldots, d$ let
  $P_\ell[I_{k-1},I_{k }]$ denote the elements on path $P_\ell$ between
  $I_{k-1}$ and $I_{k}$. Also let $R$ denote the random subset where
  each element $e$ on path $P_\ell$ is chosen independently w.p.\ $p_e$.

  For $k=1,\ldots, d$, define $J_k$ as follows:
  \begin{gather*}
    J_k = \arg\max \{ f_{L_{k-1}} (e) \mid e\in R\cap \ell[I_{k-1},I_{k
    }], f_{L_{k-1}} (e) > 0 \} \mbox{ w.p. } \frac{1}{2} \quad
    \mbox{and} \quad J_k=\nil \mbox{ w.p. } \frac{1}{2},
  \end{gather*}
  where $L_{k-1} :=\{I_1,\ldots,I_{k-1}\}\cup \{J_1,\ldots,J_{k-1}\}$.
  In words, the sets $L$ contain the exit points from the stems, and for
  each stem also  the element with maximum marginal value (if any) with
  probability half.

  For (i), by Definition~\ref{defn:proxy}, the value of $\adapbar(\T,f)$
  conditioned on path $P_\ell$ and elements $J_1,\ldots,J_d$
  is 
  \begin{equation}
    \sum_{k=1}^d f_{L_{k-1}}(I_k) + f_{L_{k-1}}(J_k) \quad \ge \quad
    \sum_{k=1}^d f_{L_{k-1}}(\{I_k,J_k\}) \quad =\quad
    f(\{I_1,J_1,\ldots I_d,J_d\}). \label{eq:ad-path-proxy} 
  \end{equation}
  The inequality follows from the following two cases:
  \begin{itemize}
  \item If $I_k\ne J_k$, then by submodularity of $f_{L_{k-1}}$,
    \begin{gather*}
      f_{L_{k-1}}(I_k) + f_{L_{k-1}}(J_k)\ge f_{L_{k-1}}(\{I_k,J_k\}) +
      f_{L_{k-1}}(\emptyset) = f_{L_{k-1}}(\{I_k,J_k\}).
    \end{gather*}
  \item If $I_k = J_k$, then by choice of $J_k$ we have
    $f_{L_{k-1}}(J_k)>0$ and
    \begin{gather*}
      f_{L_{k-1}}(I_k) + f_{L_{k-1}}(J_k) = 2\cdot
      f_{L_{k-1}}(J_k)>f_{L_{k-1}}(J_k).
    \end{gather*}
  \end{itemize}
  Using~\eqref{eq:ad-path-proxy} and taking expectation over the $J$s,
  $\adapbar(\T,f)$ conditioned on path $P_\ell$ is at least
  \begin{gather*}
    \E_{J_1,\ldots J_d}\left[ f(\{I_1,J_1,\ldots I_d,J_d\}) \right] \ge
    \frac12\cdot f(\{I_1,\ldots I_d\}).
  \end{gather*}
  Above we used Lemma~\ref{lem:BFNS} on the non-negative submodular
  function $h(S) := f(S \cup \{I_1, \ldots, I_d\})$, using the fact that
  the set $\{J_1,\ldots,J_d\}$ contains each element with probability at
  most half. Finally, deconditioning over $\ell$ (i.e., over $I_1,\ldots
  I_d$) proves part~(i).

  For part~(ii), by Definition~\ref{defn:proxy}, the value of
  $\algbar(\T,g)$ conditioned on path $P_\ell$ and elements
  $J_1,\ldots,J_d$ is
  \begin{equation*}
    \sum_{k=1}^d g_{L_{k-1}}(J_k)  \quad \le \quad \sum_{k=1}^d
    g_{J_1,\ldots J_{k-1}}( J_k ) \quad =\quad g(\{J_1,\ldots
    ,J_d\}), \label{eq:nonad-path-proxy} 
  \end{equation*}
  where the inequality is by submodularity of $g$. Since \alg chooses
  the maximum value subset in $R$ and $\{J_1,\ldots ,J_d\}\sse R$,
  taking expectations over $\ell$ and $R$, we prove part~(ii).
\end{proof}

\begin{lemma}
  \label{lem:algbar-adapbar}
  For any strategy tree $\T$ and submodular function $g$ with
  $g(\emptyset)=0$, $\algbar(\T,g)\ge \frac15\cdot \adapbar(\T,g)$.
\end{lemma}
\begin{proof}
  We proceed by induction. Recall the notation in
  Definition~\ref{defn:proxy}. For each node $i$ on the stem of \T
  define $a_i := \max\{g(i),0\}$. Note that $g(J) = a_J$ by choice of
  $J$: if $J \ne \nil$ we have $g(J)>0$ and if $J=\nil$,
  $g(J)=g(\emptyset)=0 = a_J$. We will show that
  \begin{equation}
    \label{eq:non-mon-stem} \E_{I,J}[a_I] \le 4\cdot \E_{I,J}[a_J]. 
  \end{equation}
  Then the definition of $\adapbar(\T,g)$ and $\algbar(\T,g)$, and
  induction on $\T_I$ and $g_{I\cup J}$, would prove the
  lemma. 

  Let $K = \arg\max \{ a_e \mid e\in R\}$ be the r.v.\ denoting the
  maximum weight active (i.e., in $R$) element on the stem. Then, by
  definition of $J$, we have $\E_{I,J} [a_J] = \frac12 \E_{I,K}
  [a_{K}]$. Finally we can use Lemma~\ref{lemma:stemmass} from
  Section~\ref{sec:monotone} to obtain $\E_{I,K} [a_{K}]\ge \frac12
  \E_{I,J}[a_I]$, which proves~\eqref{eq:non-mon-stem}.
\end{proof}

\newcommand{\OO}{\ensuremath{ \mathcal{O}}}

\def\smlc{\ensuremath{{\sf StocMLSC}}\xspace}

\section{Monotone XOS Functions}

In this section we study adaptivity gaps for monotone non-negative XOS
functions. To recall, a function is monotone XOS if there exist linear functions
$\C_1, \C_2, \ldots, \C_{\W}: X \to \R^+$ such that $f(S) =
\max_{i=1}^{\W} \{ \sum_{e\in S} \C_i(e) \}$. To simplify notation we
use $\C_i(S):= \sum_{e\in S} \C_i(e)$ for any $i$ and subset $S\sse X$.
The {width} of an XOS function is the smallest number $\W$ such that $f$
can be written as the maximum over $\W$ linear functions. Let $\T^*$
denote the optimal adaptive strategy.  By monotonicity $\fmax = f$ and
\eqref{eq:adap} gives
\[
  \adap(\T^*, f) = \E_{\ell \gets \pi_{\T^*}} [f(A_\ell)]. 
\]

The following is our main result in this section.
\begin{theorem}
  \label{thm:xos}
  The stochastic probing problem for monotone XOS functions of width $\W$ has
  adaptivity gap $O(\log \W)$ for any prefix-closed
  constraints. Moreover, there are instances with $W = O(n)$ and adaptivity gap
  $\Omega(\frac{\log \W}{\log\log \W})$.
\end{theorem}

In \S\ref{sec:nonadapalgo}, we also present an efficient non-adaptive
algorithm for XOS functions of width $\W$ that makes $O(\W+\log n)$ calls to the
following linear oracle.
\begin{definition}[Oracle $\OO$]
  Given a prefix-closed constraint family $\mathcal{F}$ and linear
  function $\C: X\to \R^+$, oracle $\OO(\mathcal{F}, \C)$ returns a set
  $S\in {\cal F}$ that maximizes $\sum_{e \in S} \C(e)$.
\end{definition}

\subsection{Adaptivity Gap Upper Bound}\label{subsec:xos-ub}

We first state a useful property that is used critically later.
\begin{assumption}[Subtree property]
  \label{defn:subtree} 
  For any node $u$ in the optimal adaptive strategy tree $\T^*$, if we
  consider the subtree $\T'$ rooted at $u$ then the expected value of
  $\T'$ is at most that of $\T^*$:
  \[
  \adap(\T', f) \leq \adap(\T^*, f), \text{ when $\T'$ is a subtree of
    $\T^*$}.
  \]
\end{assumption}
This is because otherwise a better strategy would be to go directly
to $u$ (probing all the element along the way, so that we satisfy the
prefix-closed constraint, but ignore these elements), and then to run
strategy $\T'$.

\paragraph{Proof Idea.}
The proof consists of three steps. In the first step we argue that one
can assume that every coefficient in every linear function $\C_i$ is
smaller than $O(\frac{\adap(\T^*, f)}{\log \W})$ (else there is a simple
non-adaptive strategy that is comparable to the adaptive value obtained
from a single active item). The second step shows that by losing a
constant factor, one can truncate the tree $\T^*$ to obtain tree $\T$,
where the instantiated value at each leaf is at most
$2\cdot \adap(\T^*, f)$. The combined benefit of these steps is to
ensure that root-leaf paths have neither high variance nor too large a
value. In the third step, we use Freedman's concentration inequality
(which requires the above properties of $\T$) to argue that for any
linear function $\C_i$, the instantiated value on a random root-leaf
path is close to its mean with high probability. Taking union bound over
the $\W$ linear functions, we can then show that (again with high
probability), no linear function has an instantiation much more than its
mean. Hence, for a random root-leaf path, $\adap$ gets value (the
maximum instantiation over linear functions) that is not much more than
the corresponding mean, which is a lower bound on the non-adaptive
value.

Below we use $\OPT=\adap(\T^*, f)$ to denote the optimal adaptive value. 

\paragraph{Small and large elements.}
Define $\lambda := 10^3\log \W$. An element $e\in X$ is called {\em
  large} if $\max_{i=1}^{\W} \C_i(e) \ge h := \frac{\OPT}{\lambda}$; it
is called {\em small} otherwise. Let $L$ be the set of large elements,
and let $\OPT_l$ (resp., $\OPT_s$) denote the value obtained by tree
$\T^*$ from large (resp., small) elements. By subadditivity, we have
$\OPT_l+\OPT_s\ge \OPT$.

Lemma~\ref{lem:adaplargecoeff} shows that that when $\OPT_l\ge \OPT/2$,
a simple non-adaptive strategy proves that the adaptivity gap is $O(\log
\W)$. Then Lemma~\ref{lem:adapsmallcoeff} shows that when $\OPT_s\ge
\OPT/2$, the adaptivity gap is $O(1)$. Choosing between the two by
flipping an unbiased coin gives a non-adaptive strategy that proves the
adaptivity gap is $O(\log W)$.  This would prove the first part of
Theorem~\ref{thm:xos}.  

\begin{lemma}\label{lem:adaplargecoeff}
  Assuming that $\OPT_l\ge \OPT/2$, there is a non-adaptive solution of
  value $\Omega(1/\log \W)\cdot \OPT$. 
Moreover, there is a solution $S$ satisfying the probing constraint with $h\cdot \min\{\sum_{e\in S\cap L} p_e,\,1\}\ge \frac{\OPT}{O(\log W)}$.  \end{lemma}
\begin{proof}
  We restrict the optimal tree $\T^*$ to the large elements. So each node in
  $\T^*$ either contains a large element, or corresponds to making a
  random choice (and adds no value). The expected value of this
  restricted tree is $\OPT_l$.  We now truncate $\T^*$ to obtain tree
  $\overline{\T}^*$ as follows. Consider the first active node $u$ on
  any root-leaf path, remove the subtree below the \yes (active) arc
  from $u$, and assign exactly a value of $h$ to this instantiation.
  The subtree property (Assumption~\ref{defn:subtree}) implies that the
  expected value in this subtree below $u$ is at most $\OPT$. On the
  other hand, just before the truncation at $u$, the adaptive strategy
  gains value of $h = \frac{\OPT}{\lambda}$ since it observed an active
  large element at node $u$. By taking expectations, we obtain that the
  value of $\overline{\T}^*$ is at least $\frac{1}{1 + \lambda} \cdot
  \OPT_l$.

  Note that $\overline{\T}^*$ is a simpler adaptive strategy. In fact
  $\overline{\T}^*$ is a feasible solution to the stochastic probing
  instance with the probing constraint ${\cal F}$ and a different
  objective $g(R) = h\cdot \min\{|R \cap L|,1\}$ which is the rank
  function of the uniform-matroid of rank~1 (scaled by $h$) over all
  large elements. As any matroid rank function is a monotone submodular
  function, Theorem~\ref{thm:monotSubmod} implies that there is a
  non-adaptive strategy which probes a feasible sequence of elements $S
  \in {\cal F}$, having value $\E_{R\sim S(\p)} [g(R)] \geq \nicefrac13
  \cdot \adap(\overline{\T}^*, g) \geq \nicefrac13 \cdot
  \frac{\OPT_l}{1+\lambda}$. Note that for any subset $R \sse L$ of
  large elements $f(R)\ge \max_{e\in R} \{ \max_{i=1}^{\W} \C_i(e)\} \ge
  h \cdot \min\{|R|,1\} = g(R)$; the first inequality is by monotonicity
  of $f$ and the second is by definition of large elements.  So we have:
  \[ \E_{R\sim S(\p)} [f(R)] \ge \E_{R\sim S(\p)} [g(R)] = \Omega(1/\log
  \W)\cdot \OPT_l. \] It follows that $S$ is the claimed non-adaptive
  solution for the original instance with objective $f$.
  
  We now show the second part of the lemma using the above solution $S$. Note that 
  $$\frac{\OPT}{O(\log W)} = \E_{R\sim S(\p)} [g(R)] = h\cdot \E_{R\sim S(\p)} [\mathbf{1}(R\cap L\ne \emptyset)] \le h\cdot \min\left\{\sum_{e\in S\cap L} p_e,\,1 \right\},$$
  as desired. 
\end{proof}

In the rest of this section we prove the following, which implies an
$O(\log \W)$ adaptivity gap.
\begin{lemma}
  \label{lem:adapsmallcoeff}
  Assuming that $\OPT_s\ge \OPT/2$, there is a non-adaptive solution of
  value $\Omega(1)\cdot \OPT$.
\end{lemma}

\begin{proof}
  We start with the restriction of the optimal tree $\T^*$ to the small
  elements; recall that $\OPT_s$ is the expected value of this
  restricted tree.
  The next step is to truncate tree $\T^*$ to yet another tree \T with
  further useful properties. For any root-leaf path in $\T^*$ drop the
  subtree below the first node $u$ (including $u$) where $f(A_u) > 2
  \cdot \OPT$; here $A_u$ denotes the set of active elements on the path
  from the root to $u$. The subtree property
  (Assumption~\ref{defn:subtree}) implies that the expected value in the
  subtree below $u$ is at most $\OPT$. On the other hand, before the
  truncation at $u$, the adaptive value obtained is more that $2\cdot
  \OPT$. Hence, the expected value of $\T^*$ obtained at or above the
  truncated nodes is at least $\frac23\cdot \OPT_s$. Finally, since all
  elements are small and thus the expected value from any truncated node
  itself is at most $h\le 0.01\cdot \OPT$, the tree $\T$ has at least
  $(\frac23 - 0.01) \OPT_s \geq \frac12 \OPT_s$ value. This implies the
  next claim:
  \begin{claim}
    \label{clm:xos-trunc}
    Tree \T has expected value at least $\frac12\cdot \OPT_s\ge
    \frac14\cdot \OPT $ and $\max_{\ell\in \T} \max_{i=1}^{\W} \{
    \C_i(P_\ell)\} \le 2\cdot \OPT$.
  \end{claim}

  Next, we want to claim that each linear function behaves like its
  expectation (with high probability) on a random path down the tree. 
  For any $i\in [\W]$ and root-leaf path $P_\ell$ in $\T$, define
  \[ \mu_i(P_\ell) := \E_{R\sim X(\p)}[\C_i(R \cap P_\ell)] = \sum_{v \in
    P_\ell} \left( p_{\elt(v)}\cdot \C_i(\elt(v)) \right).
  \]

\begin{claim}\label{clm:concxos}
  For any $i \in [\W]$,
  \begin{gather}
    \Pr_{\ell \gets \pi_{\T}}\left[|\C_i(A_\ell) - \mu_i(P_\ell)| > 0.1
      ~\OPT \right] \leq \frac{1}{\W^2}.
  \end{gather}
\end{claim}
\begin{proof}
  Our main tool in this proof is the following concentration inequality
  for martingales.
  \begin{theorem}[Freedman, Theorem 1.6
    in~\cite{Freedman}]\label{thm:freedman}
    Consider a real-valued martingale sequence $\{X_t\}_{t \geq 0}$ such
    that $X_0 = 0$, and $\E \left[ X_{t+1} \mid X_t, X_{t-1}, \ldots,
      X_0 \right] = 0$ for all $t$. Assume that the sequence is
    uniformly bounded, i.e., $| X_t | \leq M$ almost surely for all
    $t$. Now define the predictable quadratic variation process of the
    martingale to be
    $W_t = \sum_{j=0}^{t} \E \left[ X^2_{j} \, \mid \, X_{j-1}, X_{j-2},
      \ldots, X_0 \right]$
    for all $t \geq 1$. Then for all $\ell \geq 0$ and $\sigma^2 > 0$,
    and any stopping time $\tau$ we have
    \[
    \Pr \left[ \big|\sum_{j=0}^\tau X_{j}\big| \geq \ell \, \, and \, \, W_{\tau}
      \leq \sigma^2 \right] \quad\leq \quad 2
    \exp\left(-\frac{\ell^2/2}{\sigma^2 + M\ell/3} \right). \]
  \end{theorem}

  Consider a random root-leaf path $P_\ell = \langle r=v_0,v_1,\ldots,
  v_\tau=\ell \rangle$ in $\T$, and let $e_t = \elt(v_t)$. Now define a
  sequence of random variables $X_0, X_1, \ldots,$ where
  \[ X_t = \left( \one{e_t \in A} - p_{e_t}\right) \cdot \C_i(e_t).
  \]
  Let $\calH_t$ be a filter denoting the sequence of variables before
  $X_t$. Observe that $\E[X_t \mid \calH_t] = 0$, which implies $\{X_t
  \}$ forms a martingale. Clearly $|X_t| \leq |\C_i(e_t)| \leq h$. Now,
  \begin{align*}
    \sum_{j=0}^{t} \E \left[ |X_{j}| \, \mid \, \calH_j \right] &\leq
    \sum_{j=0}^{t} \left( p_{e_t}(1-p_{e_t}) + (1-p_{e_t})p_{e_t}\right)
    \cdot \C_i(e_t) \\ &\leq \frac12 \sum_{j=0}^{t} \C_i(e_t) \le
    \frac12 \cdot \max_\ell \max_{i=1}^{\W} \C_i(P_\ell) \leq {\OPT},
  \end{align*}
  where the last inequality is by Claim~\ref{clm:xos-trunc}.  We use
  $|X_j| \leq h = \frac{\OPT}{\lambda}$ and the above equation to bound
  the variance,
  \[
  \sum_{j=0}^{t} \E \left[ X^2_{j} \, \mid \, \calH_j \right] \leq
  h\cdot \sum_{j=0}^{t} \E \left[ |X_{j}| \, \mid \, \calH_j \right]
  \leq \frac{ \OPT^2}{\lambda}.
  \]
  Applying Theorem~\ref{thm:freedman}, we get
  \begin{align*}
    \Pr\left[ \big| \sum_{j=0}^\tau X_j \big| > 0.1 ~\OPT \right] &= \Pr[ |\C_i(A_\ell) -
    \mu_i(P_\ell)| >  0.1 ~\OPT ]\\
    &\leq 2 \exp\left(-\frac{(0.1 ~\OPT)^2/2}{ \OPT^2/ (\lambda) \,+\,
        (\OPT/ (\lambda)) \cdot (0.1 ~\OPT)/3} \right) \\
    &\leq \frac{1}{{\W}^2}.
  \end{align*}
  This completes the proof of Claim~\ref{clm:concxos}. 
\end{proof}

Now we can finish the proof of Lemma~\ref{lem:adapsmallcoeff}.  We label
every leaf $\ell$ in $\T$ according to the linear function $\C_i$ that
achieves the value $f(A_\ell)$, breaking ties arbitrarily. I.e., for
leaf $\ell$ we define
\[ c^{\max}_\ell := \C_i, \text{ where } \C_i(A_\ell) = f(A_\ell).
\]
Also define $\mu^{\max}_\ell := \mu_i$ for $i$ as above.
Using Claim~\ref{clm:concxos} and taking a union bound over all $i\in
[\W]$, 
\begin{equation}
  \label{eq:xos-inst-mean}
  \Pr_{\ell\gets \pi_\T} \bigg[\big| c^{\max}_\ell (A_\ell) -
  \mu^{max}_\ell (P_\ell) \big| > 0.1 ~\OPT \bigg] ~~\le~~ \frac1{\W}.
\end{equation}
Consider the natural non-adaptive solution which selects $\ell\gets
\pi_\T$ and probes all elements in $P_\ell$. This has expected value at
least:
\begin{gather*}
  \E_{\ell\gets \pi_\T} \left[\mu^{max}_\ell (P_\ell)\right]
  ~~\stackrel{\small (\ref{eq:xos-inst-mean})}{\ge}~~ \E_{\ell\gets
    \pi_\T} \left[c^{max}_\ell (A_\ell)\right] - 0.1 ~\OPT -
  \nicefrac{1}{\W}(2~\OPT) ~~\stackrel{\small
    (\text{Claim~\ref{clm:xos-trunc}})}{\ge}~~ (0.15-\nicefrac{2}{\W})\cdot \OPT.
\end{gather*}
If $\W\ge 20$ then we obtain the desired non-adaptive strategy. The
remaining case of $\W<20$ is trivial: the adaptivity gap is~$1$ for a
single linear function, and taking the best non-adaptive solution among
the $\W$ possibilities has value at least $\frac{1}{\W}\cdot
\OPT$. This completes the proof of Lemma~\ref{lem:adapsmallcoeff}.
\end{proof}

Let us record an observation that will be useful for the non-adaptive
algorithm. 
\begin{remark} 
  \label{rem:smallcoeff}
  Observe that the above proof shows that when $\OPT_s\ge \OPT/2$, 
  there exists a path $Q$ in $\T^*$  (i.e. $Q$ satisfies the probing constraints) and a linear function $\C_j$ with mean
  value $\E_{R\sim Q(\p)} [\C_j(R)] = \Omega(\OPT)$.
\end{remark}

\subsection{Polynomial Time Non-adaptive Algorithm } \label{sec:nonadapalgo}

Consider any instance of the stochastic probing problem with a width-$W$ monotone XOS objective and prefix-closed constraint ${\cal F}$. 
Our non-adaptive algorithm is the following (here $\lambda=10^3\log W$ is as in \S\ref{subsec:xos-ub}).

\begin{algorithm} [h!]
\begin{algorithmic}[1]
  \caption{Non-adaptive Algorithm for XOS functions}
  \State \textbf{define} $ m := \max_{e \in X} \{ p_e \cdot \max_{i\in [\W]} \C_i(e) \}$
  \For {$j \in \{0, \dots, 1+\log n\}$}
     \State \textbf{define} $\mathbf{b}_j$ with $\mathbf{b}_j (e) = p_e$ if
         $\max_{i\in [\W]} \{\C_i(e)\} \geq \frac{2^j m}{\lambda}$ and $\mathbf{b}_j (e) = 0$ otherwise. 
     \State  $T_j \gets \OO(\mathcal{F},\mathbf{b}_j)$ and $v(T_j)\gets \frac{2^j m}{\lambda}\cdot \min\{\mathbf{b}_j (T_j),1 \}$.
  \EndFor
  \For {$i \in \{1, \dots, \W\}$}
     \State \textbf{define} $\mathbf{c}_i$ with $\mathbf{c}_i (e) = p_e\cdot \C_i(e)$ 
          \State  $S_i \gets \OO(\mathcal{F},\mathbf{c}_i)$ and $v(S_i) \gets \mathbf{c}_i(S_i)$.
  \EndFor
  \State \Return set $S \in \{S_1, \dots, S_W, T_0, T_1, \dots,
  T_{1+\log n} \} $ that maximizes $v(S)$. 
\end{algorithmic}
\end{algorithm}

\paragraph{Case I: $\OPT_l \geq \OPT/2$.}
Lemma~\ref{lem:adaplargecoeff} shows that in this case it suffices to
consider only the set of large elements and to maximize the probability
of selecting a single large element. While we do not know $\OPT$, and
the large elements are defined in terms of $\OPT$, we do know $m =
\max_{e \in X} \{ p_e \cdot \max_{i\in [\W]} \C_i(e) \} \leq \OPT \leq
n\cdot m$.  In the above algorithm, consider the value of $j \in \{0,
\dots, 1+\log n\}$ when $2^j \cdot m/\lambda$ is between $h$ and
$2h$. Let $L$ denote the set of large elements; note that these
correspond to the elements with positive $\mathbf{b}_j(e)$ values.  By
the second part of Lemma~\ref{lem:adaplargecoeff}, the solution $T_{j}$
returned by the oracle will satisfy $v(T_j) \geq \OPT/ O(\log \W)$. Now
interpreting this solution $T_j$ as a non-adaptive solution, we get an
expected value at least:
\begin{align*}
&h\cdot \E_{R\sim T_j(\p)}[\mathbf{1}(R\cap L\ne \emptyset)] \quad = \quad h\cdot \left(1 - \Pi_{e\in T_j} (1-\mathbf{b}_j(e))\right) \quad \ge \quad h\cdot \left(1 - e^{-\mathbf{b}_j(T_j)}\right) \\
&\ge (1-1/e)h\cdot \min\{\mathbf{b}_j(T_j),1 \} \quad = \quad (1-1/e) \cdot v(T_j) \quad \ge \quad \frac{\OPT}{O(\log W)}
\end{align*}

\paragraph{Case II: $\OPT_s \geq \OPT/2$.}
In this case Remark~\ref{rem:smallcoeff} following the proof of
Lemma~\ref{lem:adapsmallcoeff} shows that  there exists a solution $Q$ satisfying the probing constraints ${\cal F}$ and a linear function $\C_j$ with mean
value $\mathbf{c}_j(Q) = \E_{R\sim Q(\p)} [\C_j(R)] = \Omega(\OPT)$. Since the above algorithm calls
$\OO(\mathcal{F},\mathbf{c}_i)$ for each $i\in [W]$ and chooses the best one, it will return a set with value $\Omega(\OPT)$.

\subsection{Adaptivity Gap Lower Bound}
\label{sec:xos-lower-bd}

Consider a $k$-ary tree of depth $k$, whose edges are the ground set.
Each edge/element has probability $p_e = \frac1k$. Here, imagine $k =
\Theta(\frac{\log n}{\log \log n})$, so that the total number of edges
is $\sum_{i = 1}^k k^i = n$. For each of the $k^k$ leaves $l$, consider
the path $P_l$ from the root to that leaf. The XOS function is $f(S) :=
\max_l |P_l \cap S|$. Note that the width $\W = \Theta(n)$ in this case.

Suppose the probing constraint is the following prefix-closed
constraint: there exists a root-leaf path $P_l$ such that all probed
edges have at least one endpoint on this path. This implies that we can
probe at most $k^2$ edges.
\begin{itemize}
\item For an adaptive strategy, probe the $k$ edges incident to the
  root. If any one of these happens to be active, start probing the $k$
  edges at the next level below that edge. (If none were active, start
  probing the edges below the left-most child, say.)  Each level will
  have at least one active edge with probability $1- (1 - \frac1k)^k
  \geq 1 - 1/e$, so we will get an expected value of $\Omega(k)$.
\item Now consider any non-adaptive strategy: it is specified by the
  path $P_l$ whose vertices hit every edge that is probed. There are
  $k^2$ such edges, we can probe all of them. But the XOS function can
  get at most $1$ from an edge not on $P_l$, and it will get at most $k
  \cdot 1/k = 1$ in expectation from the edges on $P_l$.
\end{itemize}
This shows a gap of $\Omega(k) = \Omega(\frac{\log n}{\log \log n})$ for
XOS functions with a prefix-closed (in fact subset-closed) probing
constraint. 

\subsubsection{A Lower Bound for Cardinality Constraints}

We can show a near-logarithmic lower bound for XOS functions even for
the most simple cardinality constraints. The setup is the same as above,
just the constraint is that a subset of at most $k^2$ edges can be probed.

\begin{itemize}
\item The adaptive strategy remains the same, with expected value
  $\Omega(k)$.
\item We claim that any non-adaptive strategy gets expected value
  $O(\log k)$. Such a non-adaptive strategy can fix any set $S$ of $k^2$
  edges to probe. For each of these edges, choose an arbitrary root-leaf
  path passing through it, let $T$ be the edges lying in these $k^2$
  many root-leaf paths of length $k$. So $|T| \leq k^3$. Let us even
  allow the strategy to probe all the edges in $T$---clearly this is an
  upper bound on the non-adaptive value.

  The main claim is that the expected value to be maximized when $T$
  consists of $k^2$ many disjoint paths. (The $k$-ary tree does not have
  these many disjoint paths, but this is just a thought-experiment.) The
  claim follows from an inductive application of the following simple fact.

  \begin{fact}
    Given independent non negative random variables $X,X',Y,Z$, where
    $X'$ and $X$ have the same distribution, the following holds: 
    \[ \E_{X,Y,Z} [\max\{X + Y, X + Z\}] \leq \E_{X,X',Y,Z} [\max\{X +
    Y, X' + Z\}]. \]
  \end{fact}
  \begin{proof}
    Follows from the fact that $\{\max\{X + Y, X + Z\} > c\} \sse \{
    \max\{X + Y, X' + Z\} > c\}$.
  \end{proof}

  Finally, for any path with $k$ edges, we expect to get value $1$ in
  expectation. The probability that any one path gives value $c \log k$
  is $\frac{1}{k^3}$, for suitable constant $c$. So a union bound
  implies that the maximum value over $k^2$ path is at most $c \log k$
  with probability $1/k$. Finally, the XOS function can take on value at
  most $k$, so the expected value is at most $1 + c \log k$.
\end{itemize}

This shows an adaptivity gap of $\Omega(\frac{k}{\log k}) = 
\Omega(\frac{\log n}{(\log \log n)^2})$ even for cardinality constraints.

\section{Conclusions}
\label{sec:conclusions}

In this paper we saw that submodular functions, both monotone and
non-monotone, have a constant adaptivity gap, with respect to all
prefix-closed probing constraints. Moreover, for monotone XOS functions
of width $W$, the adaptivity gap is $O(\log W)$, and there are
nearly-matching lower bounds for all $W = O(n)$. 

The most obvious open question is whether for \emph{all} XOS functions,
the adaptivity gap is $O(\log^c n)$ for some constant $c \geq 1$. This
would immediately imply an analogous result for all \emph{subadditive}
functions as well. (In \S\ref{sec:non-monotone-xos} we show that it
suffices to bound the adaptivity gap for \emph{monotone} XOS and
subadditive functions.)

Other questions include: can we get better bounds for special submodular
functions of interest? E.g., for matroid rank functions, can we improve
the bound of $3$ from Theorem~\ref{thm:monotSubmod}. We can improve the
constants of $40$ for the non-monotone case with more complicated
analyses, but getting (near)-tight results will require not losing the
factor of $4$ from~(\ref{eq:fmv}), and may require a new insight.  Or
can we do better for special prefix-closed constraints. Our emphasis was
to give the most general result we could, but it should be possible to
do quantitatively better for special cases of interest.

%%%%%%%%%%%%%%%%%%%%%%%%%%%%%%%%%%%%%%%%%%%%

%\section*{Acknowledgements}

%%%%%%%%%%%%%%%%%%%%%%%%%%%%%%%%%%%%%%%%%%%%

\bibliographystyle{alpha}
{\small \bibliography{probe}}

\newcommand{\etalchar}[1]{$^{#1}$}
\begin{thebibliography}{GKMR11}

\bibitem[Ada11]{A11}
Marek Adamczyk.
\newblock Improved analysis of the greedy algorithm for stochastic matching.
\newblock {\em Inf. Process. Lett.}, 111(15):731--737, 2011.

\bibitem[AGM15]{AGM15}
Marek Adamczyk, Fabrizio Grandoni, and Joydeep Mukherjee.
\newblock Improved approximation algorithms for stochastic matching.
\newblock {\em CoRR}, abs/1505.01439, 2015.

\bibitem[AN16]{AN16}
Arash Asadpour and Hamid Nazerzadeh.
\newblock Maximizing stochastic monotone submodular functions.
\newblock {\em Management Science}, 2016.
\newblock to appear.
  \url{http://www-bcf.usc.edu/~nazerzad/pdf/stochastic_submodular.pdf}.

\bibitem[ANS08]{ANS}
Arash Asadpour, Hamid Nazerzadeh, and Amin Saberi.
\newblock Stochastic submodular maximization.
\newblock In {\em International Workshop on Internet and Network Economics},
  pages 477--489. Springer, 2008.
\newblock Full version appears as~\cite{AN16}.

\bibitem[AR12]{AR12}
Itai Ashlagi and Alvin~E.\ Roth.
\newblock New challenges in multihospital kidney exchange.
\newblock {\em American Economic Review}, 102(3):354--59, 2012.

\bibitem[ASW14]{ASW14}
Marek Adamczyk, Maxim Sviridenko, and Justin Ward.
\newblock Submodular stochastic probing on matroids.
\newblock In {\em STACS}, pages 29--40, 2014.

\bibitem[BCN{\etalchar{+}}15]{BCNSX15}
Alok Baveja, Amit Chavan, Andrei Nikiforov, Aravind Srinivasan, and Pan Xu.
\newblock Improved bounds in stochastic matching and optimization.
\newblock In {\em APPROX}, pages 124--134, 2015.

\bibitem[BDF{\etalchar{+}}12]{BDFKNR-SODA12}
Ashwinkumar Badanidiyuru, Shahar Dobzinski, Hu~Fu, Robert Kleinberg, Noam
  Nisan, and Tim Roughgarden.
\newblock Sketching valuation functions.
\newblock In {\em Proceedings of the twenty-third annual ACM-SIAM symposium on
  Discrete Algorithms}, pages 1025--1035. SIAM, 2012.

\bibitem[BFNS14]{BFNS-SODA14}
Niv Buchbinder, Moran Feldman, Joseph~Seffi Naor, and Roy Schwartz.
\newblock {Submodular maximization with cardinality constraints}.
\newblock In {\em Proceedings of the Twenty-Fifth Annual ACM-SIAM Symposium on
  Discrete Algorithms}, pages 1433--1452. SIAM, 2014.

\bibitem[BGK11]{BGK11}
Anand Bhalgat, Ashish Goel, and Sanjeev Khanna.
\newblock Improved approximation results for stochastic knapsack problems.
\newblock In {\em SODA}, pages 1647--1665, 2011.

\bibitem[BGL{\etalchar{+}}12]{BGLMNR12}
Nikhil Bansal, Anupam Gupta, Jian Li, Juli{\'a}n Mestre, Viswanath Nagarajan,
  and Atri Rudra.
\newblock {When LP Is the Cure for Your Matching Woes: Improved Bounds for
  Stochastic Matchings}.
\newblock {\em Algorithmica}, 63(4):733--762, 2012.

\bibitem[BH11]{BalcanHarvey-STOC11}
Maria-Florina Balcan and Nicholas~JA Harvey.
\newblock Learning submodular functions.
\newblock In {\em Proceedings of the forty-third annual ACM symposium on Theory
  of computing}, pages 793--802. ACM, 2011.

\bibitem[BN14]{BN14}
Nikhil Bansal and Viswanath Nagarajan.
\newblock On the adaptivity gap of stochastic orienteering.
\newblock In {\em IPCO}, pages 114--125, 2014.

\bibitem[CCPV11]{CCPV11}
Gruia C{\u{a}}linescu, Chandra Chekuri, Martin P{\'{a}}l, and Jan
  Vondr{\'{a}}k.
\newblock Maximizing a monotone submodular function subject to a matroid
  constraint.
\newblock {\em {SIAM} J. Comput.}, 40(6):1740--1766, 2011.

\bibitem[CIK{\etalchar{+}}09]{CIKMR09}
Ning Chen, Nicole Immorlica, Anna~R. Karlin, Mohammad Mahdian, and Atri Rudra.
\newblock {Approximating Matches Made in Heaven}.
\newblock In {\em ICALP (1)}, pages 266--278, 2009.

\bibitem[CP05]{CP05}
Chandra Chekuri and Martin P{\'{a}}l.
\newblock A recursive greedy algorithm for walks in directed graphs.
\newblock In {\em FOCS}, pages 245--253, 2005.

\bibitem[DGV05]{DGV05}
Brian~C. Dean, Michel~X. Goemans, and Jan Vondr{\'a}k.
\newblock Adaptivity and approximation for stochastic packing problems.
\newblock In {\em SODA}, pages 395--404, 2005.

\bibitem[DGV08]{DGV04}
Brian~C. Dean, Michel~X. Goemans, and Jan Vondr{\'a}k.
\newblock Approximating the stochastic knapsack problem: the benefit of
  adaptivity.
\newblock {\em Math. Oper. Res.}, 33(4):945--964, 2008.

\bibitem[DHK14]{DHK14}
Amol Deshpande, Lisa Hellerstein, and Devorah Kletenik.
\newblock Approximation algorithms for stochastic boolean function evaluation
  and stochastic submodular set cover.
\newblock In {\em Proceedings of the Twenty-Fifth Annual {ACM-SIAM} Symposium
  on Discrete Algorithms, {SODA} 2014, Portland, Oregon, USA, January 5-7,
  2014}, pages 1453--1466, 2014.

\bibitem[Dob07]{Dobzinski07}
Shahar Dobzinski.
\newblock Two randomized mechanisms for combinatorial auctions.
\newblock In {\em Approximation, Randomization, and Combinatorial Optimization.
  Algorithms and Techniques, 10th International Workshop, {APPROX} 2007, and
  11th International Workshop, {RANDOM} 2007, Princeton, NJ, USA, August 20-22,
  2007, Proceedings}, pages 89--103, 2007.

\bibitem[Fei09]{Feige-SICOMP09}
Uriel Feige.
\newblock On maximizing welfare when utility functions are subadditive.
\newblock {\em SIAM Journal on Computing}, 39(1):122--142, 2009.

\bibitem[FMV11]{FMV-SICOMP11}
Uriel Feige, Vahab~S Mirrokni, and Jan Vondr{\'{a}}k.
\newblock {Maximizing non-monotone submodular functions}.
\newblock {\em SIAM Journal on Computing}, 40(4):1133--1153, 2011.

\bibitem[FNS11]{FNS-FOCS11}
Moran Feldman, Joseph Naor, and Roy Schwartz.
\newblock A unified continuous greedy algorithm for submodular maximization.
\newblock In {\em Foundations of Computer Science (FOCS), 2011 IEEE 52nd Annual
  Symposium on}, pages 570--579. IEEE, 2011.

\bibitem[FNW78]{FNW78}
M.L. Fisher, G.L. Nemhauser, and L.A. Wolsey.
\newblock {An analysis of approximations for maximizing submodular set
  functions II}.
\newblock {\em Mathematical Programming Study}, 8:73--87, 1978.

\bibitem[Fre75]{Freedman}
David~A. Freedman.
\newblock On tail probabilities for martingales.
\newblock {\em Annals of Probability}, 3:100--118, 1975.

\bibitem[GKMR11]{GKMR11}
Anupam Gupta, Ravishankar Krishnaswamy, Marco Molinaro, and R.~Ravi.
\newblock Approximation algorithms for correlated knapsacks and non-martingale
  bandits.
\newblock In {\em FOCS}, pages 827--836, 2011.

\bibitem[GKNR12]{GKNR12-soda}
Anupam Gupta, Ravishankar Krishnaswamy, Viswanath Nagarajan, and R.~Ravi.
\newblock Approximation algorithms for stochastic orienteering.
\newblock In {\em SODA}, 2012.

\bibitem[GM07]{GM07}
Sudipto Guha and Kamesh Munagala.
\newblock Approximation algorithms for budgeted learning problems.
\newblock In {\em STOC}, pages 104--113. 2007.
\newblock Full version as: \emph{Approximation Algorithms for Bayesian
  Multi-Armed Bandit Problems}, \url{http://arxiv.org/abs/1306.3525}.

\bibitem[GM09]{GuhaM09}
Sudipto Guha and Kamesh Munagala.
\newblock Multi-armed bandits with metric switching costs.
\newblock In {\em ICALP}, pages 496--507, 2009.

\bibitem[GN13]{GN13}
Anupam Gupta and Viswanath Nagarajan.
\newblock A stochastic probing problem with applications.
\newblock In {\em IPCO}, pages 205--216, 2013.

\bibitem[GNS16]{GNS16}
Anupam Gupta, Viswanath Nagarajan, and Sahil Singla.
\newblock Algorithms and adaptivity gaps for stochastic probing.
\newblock In {\em Proceedings of the Twenty-Seventh Annual ACM-SIAM Symposium
  on Discrete Algorithms}, pages 1731--1747. SIAM, 2016.

\bibitem[HKL15]{HKL15}
Lisa Hellerstein, Devorah Kletenik, and Patrick Lin.
\newblock Discrete stochastic submodular maximization: Adaptive vs.
  non-adaptive vs. offline.
\newblock In {\em Algorithms and Complexity - 9th International Conference,
  {CIAC} 2015, Paris, France, May 20-22, 2015. Proceedings}, pages 235--248,
  2015.

\bibitem[KKT15]{KKT15}
David Kempe, Jon Kleinberg, and {\'E}va Tardos.
\newblock Maximizing the spread of influence through a social network.
\newblock {\em Theory of Computing}, 11(4):105--147, 2015.

\bibitem[Kra13]{Krause13}
Andreas Krause.
\newblock Submodularity in machine learning and vision.
\newblock In {\em British Machine Vision Conference, {BMVC} 2013, Bristol, UK,
  September 9-13, 2013}, 2013.

\bibitem[LMNS09]{LMNS-STOC09}
Jon Lee, Vahab~S Mirrokni, Viswanath Nagarajan, and Maxim Sviridenko.
\newblock Non-monotone submodular maximization under matroid and knapsack
  constraints.
\newblock In {\em Proceedings of the forty-first annual ACM symposium on Theory
  of computing}, pages 323--332. ACM, 2009.

\bibitem[LPRY08]{LiuPRY08}
Zhen Liu, Srinivasan Parthasarathy, Anand Ranganathan, and Hao Yang.
\newblock Near-optimal algorithms for shared filter evaluation in data stream
  systems.
\newblock In {\em Proceedings of the {ACM} {SIGMOD} International Conference on
  Management of Data, {SIGMOD} 2008, Vancouver, BC, Canada, June 10-12, 2008},
  pages 133--146, 2008.

\bibitem[LY13]{LiY13}
Jian Li and Wen Yuan.
\newblock Stochastic combinatorial optimization via poisson approximation.
\newblock In {\em Symposium on Theory of Computing Conference, STOC'13, Palo
  Alto, CA, USA, June 1-4, 2013}, pages 971--980, 2013.

\bibitem[Ma14]{M14}
Will Ma.
\newblock Improvements and generalizations of stochastic knapsack and
  multi-armed bandit approximation algorithms: Extended abstract.
\newblock In {\em SODA}, pages 1154--1163, 2014.

\bibitem[RS{\"U}05]{RSU05}
Alvin~E.\ Roth, Tayfun S{\"o}nmez, and M.\~Utku {\"U}nver.
\newblock Pairwise kidney exchange.
\newblock {\em J. Econom. Theory}, 125(2):151--188, 2005.

\bibitem[WS11]{SW10}
David~P. Williamson and David~B. Shmoys.
\newblock {\em The design of approximation algorithms}.
\newblock Cambridge University Press, Cambridge, 2011.

\end{thebibliography}

\appendix

\section{Monotonicity of XOS and Subadditive Functions}
\label{sec:non-monotone-xos}

Our definition of fractionally subadditive/XOS differs from the usual
one, since it allows the function to be non-monotone. To show the
difference, here is the usual definition:

\begin{itemize}
\item A function $f$ is \emph{monotone fractionally subadditive} if
  $f(T) \leq \sum_i \alpha_i f(S_i)$ for all $\chi_T \leq \sum_i
  \alpha_i \chi_{S_i}$ with $\alpha_i \geq 0$. Note the subtle
  difference: we now place a constraint when the set $T$ is fractionally
  covered by the sets $S_i$. Note that such a function is always
  monotone: for $T \sse S$ we have $\chi_T \leq \chi_S$ and hence $f(T)
  \leq f(S)$.

  Similarly, we can define a function $f$ to be \emph{monotone XOS}
  (a.k.a.\ max-of-sums) if there exist linear functions $\C_1, \C_2,
  \ldots, \C_w: 2^X \to \R_{\ge 0}$ such that $f(X) = \max_j \{ \C_j(X)
  \}$. The difference is that we only allow non-negative coefficients in
  the linear functions.
\end{itemize}

The equivalence of these definitions is shown in~\cite{Feige-SICOMP09}.
It is also known that the class of monotone XOS functions lies between
monotone submodular and monotone subadditive functions. These same
proofs, with minor alterations, show that XOS functions are the same as
fractionally subadditive functions (according to the definitions in
\S\ref{sec:prelims}), and lie between general submodular and general
subadditive functions. 

Finally, if $f$ satisfies the XOS definition in \S\ref{sec:prelims}, and
$f$ is monotone, it also satisfies the definition above. Indeed, by
duplicating sets and dropping some elements, we can take the sets $T,
\{S_i\}$ and values $\alpha_i \geq 0$ satisfying $\chi_T \leq \sum_i
\alpha_i \chi_{S_i}$, and get sets $S_i' \sse S_i$ satisfying $\chi_T =
\sum_i \alpha_i \chi_{S'_i}$.
By the general XOS definition, we get
$f(T) \leq \sum_i \alpha_i f(S'_i)$, which by monotonicity is at most
$\sum_i \alpha_i f(S_i)$. Hence, the definitions in \S\ref{sec:prelims}
and above are consistent.

\subsection{Adaptivity Gaps for Non-Monotone Functions}

It suffices to prove the adaptivity gap conjecture for monotone XOS or 
subadditive functions, since for any XOS function $f$, the
function $\fmax$ is also XOS (shown below). Note that $\fmax$ is clearly monotone.  So we can just deal with the
monotone XOS/subadditive function $\fmax$. (We note that such a property is not true for submodular functions, i.e. $f$ being submodular  does not imply that $\fmax$ is.) 

Consider any (possibly non-monotone) XOS function $f$. We will show that
$\fmax$ is fractionally subadditive, i.e., for any $T\sse X$, $\{S_i\sse
X\}$ and $\{\alpha_i\ge 0\}$ with $\chi_T = \sum_i \alpha_i \chi_{S_i}$,
$\fmax(T) \leq \sum_i \alpha_i \fmax(S_i)$.

Consider any $T, \{S_i\}, \{\alpha_i\}$ as above. 
Let $U\sse T$ be the set achieving the maximum in $\fmax(T)$, i.e.
$\fmax(T)=f(U)$. Now consider the linear combination of the sets
$\{S_i\cap U\}$ with multipliers $\{\alpha_i\}$. We have $\chi_U =
\sum_i \alpha_i \chi_{S_i\cap U}$. So, by the fractionally subadditive
property of $f$,
    $$f(U)\le \sum_i \alpha_i\cdot f(S_i\cap U) \le \sum_i \alpha_i\cdot \fmax(S_i).$$
    The last inequality above is by definition of $\fmax$ as $S_i\cap U\sse S_i$. 
        Thus we have 
$\fmax(T) \leq \sum_i \alpha_i \fmax(S_i)$ as desired.

\section{ Large Adaptivity Gap for Arbitrary  Functions }\label{sec:adapgapbad}
Consider a monotone function $f$ on $k=\sqrt{n}$  types of items, with $k$ items of each type (total $k^2=n$ items). On any set $S$ of items, function $f$ takes value $1$ if $S$ contains at least one item of every type, and takes value $0$ otherwise. Suppose each item is active independently w.p. $1/2$ and the constraint allows us to probe at most $4k$ items. The optimal non-adaptive strategy here is to probe $4$ items of each type. This strategy has an expected value of $(\frac{15}{16})^k$. On the other hand, consider an adaptive strategy that arbitrarily orders the types and probes items of a type until it sees an active copy, and then moves to the next type. Since in expectation this strategy only probes $2$ items of a type before moving to the next, with constant probability it will see an active copy of every type within the $4k$ probes. Hence, the adaptivity gap for this example is $\Omega(16/15)^k$.

\end{document}